\documentclass{article}
\usepackage{amsfonts}
\usepackage{amsmath}
\usepackage{amssymb}
\usepackage{graphicx}
\usepackage{array}
\usepackage{multicol}
\usepackage[affil-it]{authblk}
\usepackage{tabulary}
\usepackage{longtable}
\usepackage{booktabs}
\usepackage{booktabs} % For \toprule, \midrule and \bottomrule
\usepackage{siunitx} % Formats the units and values
\usepackage{pgfplotstable} % Generates table from .csv

\newtheorem{theorem}{Theorem}[section]
\newtheorem{lemma}[theorem]{Lemma}
\newtheorem{proposition}[theorem]{Proposition}
\newtheorem{corollary}[theorem]{Corollary}
\newtheorem{definition}[theorem]{Definition}

\newenvironment{proof}[1][Proof]{\begin{trivlist}
\item[\hskip \labelsep {\bfseries #1}]}{\end{trivlist}}

\begin{document}

\makeatletter
\def\@maketitle{%
  \newpage
  \null
  \vskip 2em%
  \begin{center}%
  \let \footnote \thanks
    {\Large\bfseries \@title \par}%
    \vskip 1.5em%
    {\normalsize
      \lineskip .5em%
      \begin{tabular}[t]{c}%
        \@author
      \end{tabular}\par}%
    \vskip 1em%
    {\normalsize \@date}%
  \end{center}%
  \par
  \vskip 1.5em}
\makeatother

\title{On Euclidean and Hermitian Self-Dual Cyclic Codes over $\mathbb{F}_{2^r}$}

\author{Odessa D. Consorte%
  \thanks{Electronic address: \texttt{oconsorte@math.upd.edu.ph}}}
\affil{Institute of Mathematics\\ University of the Philippines, Diliman, Quezon City, Philippines}

\author{Lilibeth D. Valdez%
  \thanks{Electronic address: \texttt{ldicuangco@math.upd.edu.ph}; Corresponding author}}
\affil{Institute of Mathematics\\ University of the Philippines, Diliman, Quezon City, Philippines}

\date{}

\maketitle
\begin{abstract}
Cyclic and self-dual codes are important classes of codes in coding theory.  Jia, Ling and Xing \cite{Jia} as well as Kai and Zhu \cite{Kai} proved that Euclidean self-dual cyclic codes of length $n$ over $\mathbb{F}_q$ exist if and only if  $n$ is even and $q=2^r$, where $r$ is any  positive integer.  For $n$ and $q$ even, there always exists an $[n, \frac{n}{2}]$ self-dual cyclic code with generator polynomial $x^{\frac{n}{2}}+1$ called the \textit{trivial self-dual cyclic code}.  In this paper we prove the existence of nontrivial self-dual cyclic codes of length $n=2^\nu \cdot \overline{n}$, where $\overline{n}$ is odd,  over  $\mathbb{F}_{2^r}$ in terms of the existence of a nontrivial splitting $(Z, X_0, X_1)$ of $\mathbb{Z}_{\overline{n}}$ by $\mu_{-1}$, where $Z, X_0,X_1$ are unions of $2^r$-cyclotomic cosets mod $\overline{n}.$  We also express the formula for the number of cyclic self-dual codes over $\mathbb{F}_{2^r}$ for each $n$ and $r$ in terms of the number of $2^r$-cyclotomic cosets in $X_0$ (or in $X_1$).

We also look at Hermitian self-dual cyclic codes and show properties which are analogous to those of Euclidean self-dual cyclic codes.  That is, the existence of nontrivial Hermitian self-dual codes over $\mathbb{F}_{2^{2 \ell}}$ based on the existence  of a nontrivial splitting $(Z, X_0, X_1)$ of $\mathbb{Z}_{\overline{n}}$ by $\mu_{-2^\ell}$, where $Z, X_0,X_1$ are unions of $2^{2 \ell}$-cyclotomic cosets mod $\overline{n}.$  We also determine the lengths at which nontrivial Hermitian self-dual cyclic codes exist and the formula for the number of Hermitian self-dual cyclic codes for each $n$.
\end{abstract}


\section{Introduction}
Cyclic codes have been widely studied and have found numerous applications in storage and communication systems due to the ease in their encoding/decoding.  Both finite fields and rings have been considered as ``alphabets" in the construction of cyclic codes.  In this study, we focus on cyclic codes over finite fields.

For a code $\mathcal{C}$ of length $n$ and dimension $k$, denoted as an $[n,k]$ code over a finite field $\mathbb{F}_q$ ($q$ a power of a prime),  its dual code, $\mathcal{C}^{\perp}$ is defined to be $\{\textbf{x} \in \mathbb{F}^n_q | \textbf{x} \cdot \textbf{c} = 0 \text{  for all  } \textbf{c} \in \mathcal{C}\}$.   $\mathcal{C}^\perp$ is an $[n,n-k]$ code.   A code is said to be Euclidean self-dual if and only if $\mathcal{C}=\mathcal{C}^\perp$.   

One method of constructing self-dual codes from cyclic codes is by extending cyclic codes whose length $n$ and the characteristic of the field $\mathbb{F}_q$ are relatively prime.  Smid \cite{Smid} showed that if an extended cyclic code is self-dual, then this cyclic code is a duadic code with splitting given by $\mu_{-1}$.  Ocampo \cite{Oca} has generalized this to group codes (of which cyclic codes are a subclass) and similarly showed that for a group code $\mathcal{C}$, an ideal in the group ring $\mathbb{F}_{q^2}[G^*]$, the extended group code $\widetilde{\mathcal{C}}$  is Euclidean self-dual if and only if $\mathcal{C}$ is a split group code for some splitting $(Z=\{0\}, X_0, X_1)$ of $G$ by $\mu_{-1}.$

The other method for constructing self-dual cyclic codes considers cyclic codes of even length.  In 1983, Sloane \cite{Sloane} noticed that extensive research have been done regarding self-dual codes but cyclic self-dual codes in particular have not received much attention.  He showed that the number of distinct cyclic self-dual binary codes of length $2^ab$ ($b$ odd) depends on the number of pairs of asymmetric cyclotomic cosets modulo $b$.  Jia et al. \cite {Jia} and Kai et al. \cite{Kai} generalized this to cyclic codes over $\mathbb{F}_q$.   They proved that  self-dual cyclic codes of length $n$ over $\mathbb{F}_q$ exist if and only if $q$ is power of 2 and $n$ is even.   In particular, $n$ and $q$ are not relatively prime.  This condition gives rise to repeated-root cyclic codes.  
 
Repeated-root cyclic codes were first studied by Castagnoli et al. \cite{Cas} and van Lint \cite{Van}.  van Lint focused on binary cyclic codes of length $2n$ ($n$ odd) obtained by the $|u|u+v|$ construction.  He showed that using this construction, an infinite sequence of optimal cyclic codes with distance 4 can be obtained.  Furthermore, these codes require low complexity decoding methods.  On the other hand, Castagnoli et al. derived a parity check matrix and gave a formula for the minimum distance of repeated-root cyclic codes.  They were able to find several repeated-root binary cyclic codes that contain the maximum number of codewords among all known binary codes of the same length and minimum distance.  However, they have also demonstrated that repeated root cyclic codes are not better than general cyclic codes of the same length.  Inspite of this, it is still worthwhile to study repeated-root cyclic codes under which self-dual cyclic codes of even length are classified.

This paper is organized as follows.  In Section 2, we give the notations that will be used in this paper and review some basic concepts regarding cyclic codes.  The reader is referred to Huffman and Pless \cite{HP} for a more detailed discussion of cyclic codes.  In Section 3, we focus on Euclidean self-dual cyclic codes over $\mathbb{F}_{2^r}$.  We look at previous results and use  $q$-cyclotomic cosets mod $\overline{n}$ and splittings of $\mathbb{Z}_{\overline{n}}$ to determine the existence of nontrivial self-dual cyclic codes.  We then obtain analogous statements in Section 4 for Hermitian self-dual cyclic codes over $\mathbb{F}_{2^{2 \ell}}$.

\section{Notations and Basic Concepts}
A cyclic code $\mathcal{C}$ is a linear code of length $n$ over a field $\mathbb{F}_q$, where $q$ is a power of a prime such that if $\textbf{c}=c_0c_1 \cdots c_{n-1} \in \mathcal{C}$, then its right cyclic shift $\textbf{c}=c_{n-1}c_0 \cdots c_{n-2}$ is also an element of $\mathcal{C}$.  We consider the bijective correspondence between vectors $\textbf{c}=c_0c_1 \cdots c_{n-1}$ in $\mathbb{F}_q^n$ and polynomials $c(x)=c_0+c_1+ \cdots +c_{n-1}$ in $\mathbb{F}_q[x]$ of degree at most $n-1$.  The codeword $\textbf{c}$ cyclically shifted one to the right can be represented by $x\cdot c(x)=c_{n-1}+c_0x + \cdots + c_{n-2}x^{n-1}$ where we have set $x^n=1.$  Hence, the study of cyclic codes is equivalent to the study of the residue class ring $\mathbb{F}_q[x]/(x^n-1)$.  The study of ideals in this residue class ring hinges on factoring $x^n-1.$  Following the result of Jia et al. (Theorem 1 of \cite{Jia}), we shall consider cyclic codes of even length throughout this paper.  That is, let  $n=2^\nu\cdot \overline{n}$, where $\overline{n}$ is odd, and $\nu$ is a positive integer.

To obtain the irreducible factors of $x^n-1$,  we recall the concept of $q$-cyclotomic cosets.  Let $a$ be a non-negative integer,  $0 \leq a < \overline{n}$, and $ gcd (q,\overline{n})=1$,  the \textbf{$q$-cyclotomic coset of $a$ modulo $\overline{n}$} is the set
$$C_a=\{a, aq, aq^2, \dots, aq^{r-1}\}$$
where each element is computed modulo $\overline{n}$, $r$ is the smallest positive integer such that $aq^r \equiv a$ mod $\overline{n}$, and $a$ is usually taken as the smallest number in the set.  Note that the distinct $q$-cyclotomic cosets mod $\overline{n}$ partition the set $\{0,1, \dots,  \overline{n}-1\}$.  

Suppose $t=ord_{\overline{n}}(q)$, i.e. $t$ is the smallest positive integer such that $q^t \equiv 1$ mod $\overline{n}$.  If $\alpha$ is a primitive $\overline{n}^{th}$ root of unity in $\mathbb{F}_{q^t}$, then the \textbf{minimal polynomial} of $\alpha^a$ over $\mathbb{F}_q$, is $f_a(x)=\prod_{i \in C_a}(x-\alpha^i)$.   For $\overline{n}$ and $q$ relatively prime,  the factorization of $x^{\overline{n}}-1$ into pairwise-irreducible polynomials over $\mathbb{F}_q$ is given by
$$x^{\overline{n}}-1 = \prod_{a \in I}f_a(x),$$
where $I$ is the complete set of $q$-cyclotomic coset representatives modulo $\overline{n}$.  Hence
$$x^n-1=[x^{\overline{n}}-1]^{2^\nu} = \prod _{a \in I} f_a(x)^{2^\nu}$$
has $\overline{n}$ distinct roots with multiplicity $2^\nu$ in its splitting field.

An $[n,k]$ cyclic code can be described by its generator polynomial.  The \textbf{generator polynomial} of a code is the unique monic polynomial of degree $n-k$ that is a divisor of $x^n-1$ in $\mathbb{F}_q[x]$.   It is a product of the minimal polynomials of $\alpha^a$ where $a$ is any element of $I$ (defined above) over $\mathbb{F}_q$.  

Suppose $\mathcal{C}$ is an $[n,k]$ cyclic code with generator polynomial $g(x)$.  The polynomial
$$p(x)=\dfrac{x^n-1}{g(x)}=\sum_{i=0}^k p_ix^i$$
is called the \textit{parity-check polynomial of} $\mathcal{C}$.   Consequently, the generator polynomial of $\mathcal{C}^\perp$ is defined as
$$p^*(x)=p^{-1}_0 x^kp(x^{-1}).$$

 \section{Self-Dual Cyclic Codes}
Jia et al. \cite{Jia} and Kai et al. \cite{Kai} proved that there exists at least one self-dual cyclic code of length $n$ over $\mathbb{F}_q$ if and only if $q$ is a power of 2 and $n$ is even.  This code is the $[n,\frac{n}{2}]$ trivial self-dual cyclic code and has generator polynomial $x^{\frac{n}{2}}+1$.  Our aim is to determine the conditions under which nontrivial self-dual cyclic codes exist.

Let $f(x)$ be a polynomial in $\mathbb{F}_q[x]$.  The \textbf{reciprocal polynomial} of $f(x)$  is the polynomial
$$f^*(x)=f_0^{-1} \cdot x^{deg f} \cdot f(x^{-1})=f_0^{-1}(f_k + f_{k-1}x+ \cdots + f_0x^k),$$
where $f_0$ is the constant term of the polynomial $f(x)$.  $f(x)$ is called a \textbf{self-reciprocal} polynomial if $f(x)=f^*(x)$.

It is known that if $\alpha^{a}, \alpha^{qa}, \ldots, \alpha^{q^ka}$ are the nonzero roots of $f$ in some extension field of $\mathbb{F}_q$, then $\alpha^{-a}, \alpha^{-qa}, \ldots, \alpha^{-q^ka}$ are the nonzero roots of $f^*$ in that extension field. If $f(x)$ is irreducible over $\mathbb{F}_q$, so is $f^*(x)$.  Hence, if  $f_{a}(x)= \prod_{i \in C_a}(x-\alpha^i)$ is the minimal polynomial of $\alpha^a$ over $\mathbb{F}_q$, then its reciprocal polynomial is
$$f^*_{a}(x)= \prod_{i \in C_{a}}(x-\alpha^{-i})= \prod_{i \in C_{-a}}(x-\alpha^i)=f_{-a}(x).$$  
This implies that if $C_{a}$ is the $q$-cyclotomic coset that corresponds to the minimal polynomial $f_{a}(x)$ then $C_{-a}$ is the $q$-cyclotomic coset that corresponds to its reciprocal polynomial $f^*_{a}(x)$.  Note that $C_{-a}=\mu_{-1}C_a$, where the \textbf{multiplier} $\mu_{-1}$ is defined by $i \mu_{-1} \equiv -i$ mod $\overline{n}$ for each $i$ in $\{0,1,2, \ldots, \overline{n}-1\}.$

For $n=2^\nu \cdot \overline{n}$ and $q=2^r$, where $\overline{n}, q$ are relatively prime, $x^{\overline{n}}+1$ can be written as a product of distinct irreducible polynomial factors as  \cite{Jia}

$$x^{\overline{n}}+1=f_1(x) \cdots f_s(x) h_1(x)h_1^*(x) \cdots h_t(x)h_t^*(x),$$
where the $f_i(x)$'s are monic, irreducible, self-reciprocal polynomials over $\mathbb{F}_{2^r}$ and $h_j(x),h^*_j(x)$ form a pair of reciprocal polynomials which are also monic and  irreducible  over $\mathbb{F}_{2^r}$.  Hence,

$$x^n+1=(x^{\overline{n}}+1)^{2^\nu}=f_1(x)^{2^\nu} \cdots f_s(x)^{2^\nu} h_1(x)^{2^\nu}h_1^*(x)^{2^\nu} \cdots h_t(x)^{2^\nu}h_t^*(x)^{2^\nu}.$$

Theorem 2  of \cite{Jia} states that if $x^n+1$ is factorized as above, a cyclic code of length $n$ is self-dual over $\mathbb{F}_{2^r}$ if and only if its generator polynomial is of the form
$$g(x)=f_1(x)^{2^{\nu-1}} \cdots f_s(x)^{2^{\nu-1}} h_1(x)^{\beta_1}h_1^*(x)^{2^\nu-\beta_1} \cdots h_t(x)^{\beta_t}h_t^*(x)^{2^\nu-\beta_t}$$
where $0 \leq \beta_i \leq 2^\nu$ for each $1 \leq i \leq t$. \\

Note that using this factorization, the generator polynomial of the trivial self-dual cyclic code can be written as
$$x^{\frac{n}{2}}+1=f_1(x)^{2^{\nu-1}} \cdots f_s(x)^{2^{\nu-1}} h_1(x)^{2^{\nu-1}}h_1^*(x)^{2^{\nu-1}} \cdots h_t(x)^{2^{\nu-1}}h_t^*(x)^{2^{\nu-1}}.$$

It was earlier noted that the $f_i(x)$'s,  $h_j(x)$'s, $h_j^*(x)$'s are minimal polynomials that correspond to some $q$-cyclotomic coset $C_i, C_j$, and $C_{-j}$ respectively.  Thus, aside from studying the factors of $x^n -1$, we can look at the $2^r$-cyclotomic cosets mod $\overline{n}$ in the characterization of self-dual cyclic codes.    

\begin{definition}\label{def} \cite{Dcabelian}
Let gcd$(q, \overline{n})=1$.  A \textbf{splitting} of $\mathbb{Z}_{\overline{n}}$ by the multiplier $\mu_b, b \neq 0$  is a triple $(Z, X_0, X_1) $ which satisfies the following conditions: 
\begin{enumerate}
\item $Z, X_0,X_1$ are unions of $q$-cyclotomic cosets mod $\overline{n}$ such that
 $\mathbb{Z}_{\overline{n}}=Z \cup X_0 \cup X_1$ and $Z\cap X_0 \cap X_1 = \emptyset$.
 \item $\mu_b(Z)=Z, \mu_b(X_0)=X_1$ and $\mu_b(X_1)=X_0$.
 \end{enumerate}
 We say that a splitting is trivial if $X_0$ and $X_1$ are both empty.
 \end{definition}
 
\noindent
\begin{proposition}\label{lemA}Let $\gcd(q, \overline{n})=1$.  In the factorization of $x^{\overline{n}}-1$ over $\mathbb{F}_q$, there exists at least one pair of reciprocal, monic, irreducible polynomials if and only if there exists a nontrivial splitting $(Z,X_0,X_1)$ of $\mathbb{Z}_{\overline{n}}$ by $\mu_{-1}$ and each $q$-cyclotomic coset in $Z$ is fixed set-wise by $\mu_{-1}$.
\end{proposition}
\begin{proof}
Suppose that over $\mathbb{F}_q$, $x^{\overline{n}}-1$ factors into $f_1(x) \cdots f_s(x)h_{1}(x)h^*_{1}(x)\ldots h_{\ell}(x)$  $h^*_{\ell}(x)$, where $\ell \geq 1$, the $f_k(x)'s$ are monic, irreducible, and self-reciprocal, while $h_{j}(x)$ and $h_{j}^*(x)$ for $1 \leq j \leq \ell$ comprise a pair of reciprocal polynomials which are monic and irreducible.  Let $C_1, \ldots, C_s$ be the $q$-cyclotomic cosets which correspond to the minimal polynomials $f_1(x), \ldots, f_s(x)$ respectively.  Since the $f_i(x)$'s are self-reciprocal, $\mu_{-1}C_k=C_k$ for all $k=1,\ldots, s$.   Take $Z=C_1 \cup \cdots \cup C_s$. 

On the other hand, $h^*_{j}(x)$ is the reciprocal polynomial of  $h_{j}(x)$, where $h^*_{j}(x) \neq h_{j}(x)$.  Hence, $h^*_{j}(x)=h_{-j}(x)$ and $\mu_{-1}C_{j}=C_{-j}$, $j \neq -j$.    Take $X_0=\cup_j C_{j}$ and $X_1=\cup_j C_{-j}$.  Since $j \geq 1$, we obtain a nontrivial splitting $(Z, X_0, X_1)$ of $\mathbb{Z}_{\overline{n}}$ by $\mu_{-1.}$

Conversely, suppose there exists a nontrivial splitting of $\mathbb{Z}_{\overline{n}}$ by $\mu_{-1}$.  That is,
$$\{0,1, \ldots, \overline{n}-1\}=Z \cup X_0 \cup X_1 \text{   s.t.   }  Z \cap X_0 \cap X_1 = \emptyset $$
and
$$X_0, X_1 \neq \emptyset, \mu_{-1}X_0=X_1, \mu_{-1}X_1=X_0$$
where $Z, X_0,X_1$ are unions of $q$-cyclotomic cosets modulo $\overline{n}$, and each $q$ cyclotomic coset in $Z$ is fixed set wise by $\mu_{-1}.$  Let $Z=C_{z_1} \cup \cdots \cup C_{z_k}$ for some $k \geq 1$.  Since $\mu_{-1}C_{z_i}=C_{z_i}$ for $i=1, \ldots, k$, the corresponding minimal polynomials  $f_0(x), f_{z_1}(x), \ldots, f_{z_k}(x)$ are self-reciprocal .

Since $X_0$ and $X_1$ are both non-empty, there exists at least one $q$-cyclotomic coset in $X_0$ and $X_1$.  Suppose  $X_0=\cup_jC_{j}$ for $j \geq 1$.   Then $\mu_{-1}X_0=X_1$ implies  $X_1=\cup_jC_{-j}$.  For $j \geq 1$, let $h_j(x)$ be the minimal polynomial which corresponds to $C_{j}$ and $h_{-j}(x)$ to $C_{-j}$.  But, $h_{-j}(x)=h^*_j(x)$ for each $j$.  Hence,  $h_j(x)$  and $h_{-j}(x)$ form a pair of reciprocal polynomials in the factorization of $x^{\overline{n}}-1.$ 
\hfill $\square$
\end{proof}
This proposition and the factorization condition given by Jia et al. (Theorem 2 of  \cite{Jia}) lead to the following statement.
\begin{theorem}\label{split}
 A nontrivial Euclidean self-dual cyclic code $\mathcal{C}$  of length $n=2^\nu\cdot \overline{n}$ over $\mathbb{F}_{2^r}$ exists if and only if there exists a nontrivial splitting $(Z, X_0, X_1)$ of $\mathbb{Z}_{\overline{n}}$ by $\mu_{-1}$ where each $q$-cyclotomic coset in $Z$ is fixed set-wise by $\mu_{-1}$.
\end{theorem}

%This implies that if a cyclic code of length $n$ is self-dual then its generator polynomial is of the form:
%$$f_1(x)^{\alpha_1} \cdots f_s(x)^{\alpha_s} h_1(x)^{\beta_1}h_1^*(x)^{2^\nu-\beta_1} \cdots h_t(x)^{\beta_t}h_t^*(x)^{2^\nu-\beta_t}$$ 
%where $0 \leq  \alpha_i$ for each $1 \leq i \leq s,$ and $0 \leq \beta_j\leq 2^\nu$ for each $1 \leq j \leq t$ as was also shown by Jia et al. \cite{Jia}.    

We can then restate Corollary 1 of \cite{Jia} using $q$-cyclotomic cosets modulo $\overline{n}$ and the splitting $(Z, X_0, X_1)$ as follows. (Note that this is similar to Theorem 5  of \cite{Ned}).

\begin{corollary}
For $n=2^\nu \cdot \overline{n}, r>0$, the number of $[n, \frac{n}{2}]$ self-dual cyclic codes over $\mathbb{F}_{2^r}$  is exactly
$$(2^\nu+1)^t$$
where $t$ is the number of $2^r$-cyclotomic cosets in $X_0$ (or in $X_1$).
\end{corollary}
Next, we want to determine the lengths $n$ for which nontrivial self-dual cyclic codes over $\mathbb{F}_{2^r}$ exist.  Kai et al. \cite{Kai} has proved a similar theorem.  We will present an alternate proof using  the splittings of $\mathbb{Z}_{\overline{n}}$ by $\mu_{-1}$ described earlier.   We first state the following lemma.

\begin{lemma}\label{lemB}
For $\overline{n}$ odd, let $\mathbb{Z}_{\overline{n}}$ be partitioned into $C_0, \ldots, C_{\ell}$ where each $C_i$ is a $q$-cyclotomic coset modulo $\overline{n}$.  Then, for all $a \neq -a$ in $\{1, ..., \ell\}$, $C_a = C_{-a}$ if and only if $q^k  \equiv -1$ mod $ \overline{n} $ for any positive integer $k$.  
\end{lemma}
\begin{proof}
Suppose $a \neq -a$ and $C_a=C_{-a}$.  Then, $ \{a, qa, q^2a, \ldots, q^{m-1}a\} = \{-a, q(-a), q^2(-a), $ $\ldots, q^{m-1}(-a)\}$.  This implies that there exists $k$,  $1 \leq k \leq m-1$ such that $a \equiv q^k(-a)$ mod $\overline{n}$.  Hence,  $q^k \equiv -1$ mod $\left(\frac{\overline{n}}{gcd(a, \overline{n})}\right)$.  Taking $a=1$, we have $C_1=C_{-1}$ if and only if $q^k \equiv -1$ mod $\overline{n}$.  Then $C_j=C_{-j}$ for  $j=1, \ldots, \ell$  since $q^k \equiv -1$ mod $\overline{n}$ implies $jq^k \equiv -j$ mod $\overline{n}$. Conversely, if $q^k \equiv -1$ mod $\overline{n}$ then $aq^k \equiv -a$ mod $\overline{n}$.  Hence $C_a=C_{-a}$ for all $a=1, \ldots, \ell.$ \hfill $\square \\$
\end{proof}

\begin{theorem}\label{T2}
Nontrivial Euclidean self-dual cyclic codes of length $n=2^\nu\cdot \overline{n}$  ($\nu \in \mathbb{Z}^+, \overline{n}$ odd) over $\mathbb{F}_{2^r}, r \in \mathbb{Z}^+$ exist if and only if $2^{rk} \not \equiv -1$ mod $\overline{n}$ for all positive integers $k$.
 \end{theorem}
\begin{proof}
Suppose $\mathcal{C}$ is a nontrivial Euclidean self-dual cyclic code of length $n$ over $\mathbb{F}_{2^r}$.  By Theorem \ref{split}, there exists a nontrivial splitting $(Z, X_0, X_1)$ of $\mathbb{Z}_{\overline{n}}$ by $\mu_{-1}$ and so $X_0 = \cup_j C_j,$ for  $j \geq 1.$  By Definition \ref{def}, $\mu_{-1}(X_0)=\mu_{-1}(\cup_jC_j)=\cup_jC_{-j}=X_1$  where $X_0 \cap X_1 = \emptyset$.  Then, $ C_{j} \neq C_{-j}$ for at least one $j$.  Using Lemma \ref{lemB},  we conclude that $2^{rk} \not \equiv -1 $ mod $n$ for all $k \in \mathbb{Z}^+.$  The converse is proved similarly.  \hfill $\square$
\end{proof}

\section{Hermitian Self-Dual Cyclic Codes} 

We now consider cyclic codes over $\mathbb{F}_{q^2}$, where $q$ is a power of a prime $p$.  Let \textbf{x} $= (x_0, x_1, \ldots, , x_{n-1})$ and \textbf{y} $= (y_0, y_1, \ldots, , y_{n-1})$ be vectors in $\mathbb{F}_{q^2}^n$.  Consider the involution   $ \bar{ }: a \mapsto a^q$ defined on $\mathbb{F}_{q^2}$.  The \textit{Hermitian scalar product} of \textbf{x} and \textbf{y}
is defined to be $\textbf{x}\cdot \overline{\textbf{y}}=\sum_{i=0}^{n-1}x_i\overline{y_i}=\sum_{i=0}^{n-1}x_iy_i^q.$

Let $\mathcal{C}$ be an $[n,k]$ cyclic code over $\mathbb{F}_{q^2}^n$.  The \textbf{Hermitian dual of $\mathcal{C}$} is the set $\mathcal{C}^{\perp_H}=\{\textbf{u} \in \mathbb{F}_{q^2}^n|\textbf{u} \cdot \overline{\textbf{w}}=0 \text{ for all } \textbf{w} \in \mathcal{C}\}$.  We say that a code $\mathcal{C}$ is \textbf{Hermitian self-dual} if $\mathcal{C}=\mathcal{C}^{\perp_H}.$

We extend the involution map to polynomials in $\mathbb{F}_{q^2}[x]$. For $f(x)=f_0+f_1x+\cdots+f_{n-1}x^{n-1}$in $\mathbb{F}_{q^2}[x]$, we set $\overline{f(x)}=\overline{f_0}+\overline{f_1}x+\cdots+\overline{f_{n-1}}x^{n-1}$.  Let $f(x)$ be a polynomial in $\mathbb{F}_{q^2}[x]$ and $f^*(x)$ its reciprocal polynomial as defined earlier.  The \textit{conjugate reciprocal polynomial} of 
$f(x)$ is denoted as $f^\dag(x)$ and is equal to $\overline{f^*(x)}.$   $f(x)$ is said to be \textit{self-conjugate reciprocal} if $f(x)=f^\dag(x)$.  Otherwise, $f(x)$ and $f^\dag(x)$ form a conjugate-reciprocal pair.

Let $\mathcal{C}$ be a nonzero $[n,k]$ cyclic code over $\mathbb{F}_{q^2}$ generated by $g(x)$.  If $p(x)$ is the parity-check polynomial of $\mathcal{C}$, then the generator polynomial of  $\mathcal{C}^{\perp_H}$ is $p^\dag(x).$  Hence $\mathcal{C}$ is Hermitian self-dual if and only if $g(x)=p^\dag(x)$.

Jitman et al. \cite{Jitman} have shown that Hermitian self-dual abelian codes in $\mathbb{F}_{q^2}[G]$ (char $\mathbb{F}_{q^2}=p$; order of the finite abelian group $G=mp^k, p\nmid m$)  exist if and only if $p=2$ and $k \geq 1$.  We only consider the cyclic case i.e., Hermitian self-dual cyclic codes over $\mathbb{F}_{q^2}$  exist if and only if $q=2^{2 \ell}$ and $n$ is even.   A Hermitian self-dual cyclic code over this field of length $n$ and dimension $n/2$ can always be constructed.  Let $x^n-1=x^n+1=(x^{n/2}+1)^2$.  Take $g(x)=x^{n/2}+1$.  The parity check polynomial, $p(x)=x^{n/2}+1=p^*(x)=p^\dag(x)$.  Hence, the code generated by $g(x)=x^{n/2}+1$ is not only Euclidean self-dual but also Hermitian self-dual.   It will also be referrred to as the trivial Hermitian self-dual cyclic code.  

Consequently, in this section, we will consider codes of even length, $n= 2^\nu \cdot \overline{n}$ over $\mathbb{F}_{2^{2 \ell}}$  and characterize nontrivial Hermitian self-dual cyclic codes over this field.  Note that Dicuangco et al.\cite{Dcduadic} have proven that for a cyclic code $\mathcal{C}$ over $\mathbb{F}_{q^2}$ of odd length, the extended code is Hermitian self-dual if and only if $\mathcal{C}$ is an odd-like duadic code split by $\mu_{-q}$.   We use Theorem 3.9 in  \cite{Jitman} as a lemma to prove the existence of nontrivial Hermitian self-dual cyclic of length $n$ over $\mathbb{F}_{2^{2 \ell}}$ using splittings of $\mathbb{Z}_{\overline{n}}$. 
\begin{lemma}
(Theorem 3.9,  \cite{Jitman}) Let $n=2^\nu \cdot \overline{n}$ where $\nu$ is a positive integer and $\overline{n}$ is odd.  Over $\mathbb{F}_{2^{2 \ell}}[x]$, let $x^n+1$ be factored as
$$x^n+1= [x^{\overline{n}}+1]^{2^\nu}=\left[f_1(x) \ldots f_s(x)h_1(x)h_1^\dag(x) \ldots h_t(x) h_t ^\dag (x)\right]^{2^\nu}$$
where the $f_i's$ are monic, irreducible, self-conjugate reciprocal polynomials and $h_j, h_j^*$ which are also monic and irreducible form a conjugate-reciprocal pair for $1 \leq j \leq t$. 
A cyclic code $\mathcal{C}$ of length $n$ is Hermitian self-dual over $\mathbb{F}_{2^{2 \ell}}$ if and only if its generator polynomial is of the form
$$g(x)=f_1^{2^{\nu-1}}(x) \ldots f_s(x)^{2^{\nu-1}}h_1(x)^{\gamma_1}h_1^\dag(x)^{2^\nu-\gamma_1} \ldots h_t^{\gamma_i}(x) h_t ^\dag (x)^{2^\nu-\gamma_i} $$
where $0 \leq \gamma_i \leq 2^\nu $ for each i.
\end{lemma}

We now discuss the relationship between $2^{2 \ell}$-cyclotomic cosets and conjugate-reciprocal polynomials.  For a polynomial $f_a(x)$ in $\mathbb{F}_{2^{2 \ell}}[x]$ where $f_a(x)=\prod_{i \in C_a}(x-\alpha^i)$ and $C_a=\{a, 2^{2 \ell}a, \ldots, 2^{2 \ell r}a\}$ is the $2^{2 \ell}$-cyclotomic coset containing $a$, the nonzero roots of  $f_a(x)$ in some extension field of $\mathbb{F}_{2^{2 \ell}}$ are $\alpha^{a}, \ldots, \alpha^{2^{2 \ell r} a}.$  Using the definition of the conjugate-reciprocal polynomial and involution in $\mathbb{F}_{2^{2 \ell}}[x]$,  we can write 
$$f_a^\dag(x)=f_0^{-2^\ell}x^k\prod_{i \in C_a}(x^{-1}-\alpha^{2^\ell i})=\prod_{i \in C_a}(x-\alpha^{-2^\ell i})=\prod_{i \in C_{(-2^\ell a)}}(x-\alpha^i).$$
Hence, if the $2^{2 \ell}$-cyclotomic coset $C_a$ corresponds to the minimal polynomial  $f_a(x)$ in $\mathbb{F}_{2^{2 \ell}}$, then $C_{(-2^\ell a)}= \mu_{(-2^\ell)}C_a$ corresponds to the conjugate-reciprocal polynomial of  $f_a(x)$ which is $f_a^\dag(x)$.  

We can use this property to prove the following proposition in a manner similar to the proof of Proposition \ref{lemA} but instead of reciprocal polynomials, we use conjugate-reciprocal polynomials and instead of a splitting by $\mu_{-1}$, we consider the splitting by $\mu_{-2^\ell}.$ 

\begin{proposition}\label{propHSD}In the factorization of $x^{\overline{n}}-1$ over $\mathbb{F}_{2^{2 \ell}}$, there exists at least one pair of conjugate-reciprocal, monic, irreducible polynomials if and only if there exists a nontrivial splitting $(Z,X_0,X_1)$ of $\mathbb{Z}_{\overline{n}}$ by $\mu_{-2^\ell}$ and each $2^{2 \ell}$-cyclotomic coset in $Z$ is fixed set-wise by $\mu_{-2^\ell}$.
\end{proposition}
The following theorem which is analogous to Theorem \ref{split} can be shown to be true by using Proposition \ref{propHSD} instead of Proposition  \ref{lemA} in the proof.
\begin{theorem}
A nontrivial Hermitian self-dual cyclic code $\mathcal{C}$ of length $n=2^\nu \cdot \overline{n}$ ($\nu \in \mathbb{Z}^+, \overline{n}$  odd) over $\mathbb{F}_{2^{2 \ell}}$ exists if and only if there exists a nontrivial splitting $(Z, X_0, X_1)$ of $\mathbb{Z}_{\overline{n}}$ by $\mu_{-2^\ell}$ and each $2^{2 \ell}$-cyclotomic coset in $Z$ is fixed setwise by $\mu_{-2^\ell}.$
\end{theorem}
A corollary to this gives the number of Hermitian self-dual cyclic codes over $\mathbb{F}_{2^{2 \ell}}$.\begin{corollary}\label{CorH}
For $n=2^\nu \cdot \overline{n}$, the number of $[n, \frac{n}{2}]$ Hermitian self-dual cyclic codes over $\mathbb{F}_{2^{2 \ell}}$  is exactly
$$(2^\nu+1)^t$$
where $t$ is the number of $2^{2 \ell}$-cyclotomic cosets in $X_0$ (or in $X_1$).
\end{corollary}

We can also determine the lengths $n$ for which nontrivial Hermitian self-dual cyclic codes over $\mathbb{F}_{2^{2 \ell}}$ exist as follows.

\begin{lemma}\label{lemH}
For $q=2^\ell$ and $ \overline{n}$ odd, let $\mathbb{Z}_{\overline{n}}$ be partitioned into $C_0, \ldots, C_{j}$ where each $C_i$ is a $q^2$-cyclotomic coset modulo $\overline{n}$.  Then, for all $a \neq -a$ in $\{1, ..., j\}$, $C_a = C_{-qa}$ if and only if $q^{2k+1}  \equiv -1$ mod $\overline{n}$ for any positive integer $k$.
\end{lemma}

\begin{proof}
Let $C_0, \ldots, C_{j}$ be the $q^2$-cyclotomic cosets mod $\overline{n}$.  Suppose $C_a=C_{-qa}$.  Then, $\{a, q^2a, \ldots, q^{2(m-1)}\cdot a\}=\{-qa, -q^3a, \ldots, -q\cdot q^{2(m-1)}\cdot a\}$.  This implies that there exists $k \in \mathbb{Z}$ where $1 \leq k \leq m-1$ such that $a \equiv q^{2k} \cdot -qa$ mod $\overline{n}$.  That is,  $q^{2k+1} \equiv -1$ mod $\left(\frac{\overline{n}}{gcd(a, \overline{n})}\right)$.   If $a=1$,  $C_1=C_{-1}$ if and only if $q^{2k+1} \equiv -1$ mod $\overline{n}$.  Then $C_i=C_{-qi}$ for  $i=1, \ldots, j$  since $q^{2k+1} \equiv -1$ mod $\overline{n}$ implies $iq^{2k+1} \equiv -i$ mod $\overline{n}$. Conversely, if $q^{2k+1} \equiv -1$ mod $\overline{n}$ then $aq^{2k+1} \equiv -a$ mod $\overline{n}$.  Hence $C_a=C_{-qa}$ for all $a=1, \ldots, j.$ \hfill $\square \\$
\end{proof}
Using this lemma instead of  Lemma \ref{lemB} in the proof of Theorem \ref{T2}, we can show that the following theorem holds. 
\begin{theorem}
Nontrivial Hermitian self-dual cyclic codes of length $n=2^\nu\cdot \overline{n}$ ($\nu \in \mathbb{Z}^+, \overline{n}$  odd) over $\mathbb{F}_{2^{2 \ell}}$ exist if and only if  $2^{\ell(2k+1)} \not \equiv -1$ mod $\overline{n}$ for all positive integers $k$.
\end{theorem}

Table \ref{table1} shows the values of $n$ for which nontrivial Hermitian self-dual cyclic codes over $\mathbb{F}_4$ exist and the number of Hermitian self-dual codes (including the trivial Hermitian self-dual code) for each $n$ computed using Corollary \ref{CorH}.  Here, $n=2^\nu\cdot \overline{n}$ where $\overline{n}$ is odd and $t = $ number of 4-cyclotomic cosets in $X_0$ (or $X_1$).  The highest minimum distance (HMinD) of the cyclic code of length $n$ for $n \leq 100$ was computed using MAGMA.
  \begin{center}
  \begin{table}[h]
  \centering
  \small
    \caption{Number of Hermitian Self-Dual Cyclic Codes over $\mathbb{F}_4$}
    \label{table1}
    \begin{tabular}{ | c  c   c   c  c  c  |  c  c   c   c  c  |}
    \hline
     	$n$	&$\overline{n} $&  $\nu$	&	$t$	& \textbf{No. of HSD}	& \textbf{HMinD}& $n$	&$\overline{n}$	&  $\nu$	&	$t$	& \textbf{No. of HSD} \\	
     \hline
     	10	&	5		&	1	&	1	&	3	&  	4	&	158	&	79			&	1	&	1	&	3	\\
	14	&	7		&	1	&	1	&	3	&  	4	&	160	&	5			&	5	&	1	&	33	\\
	20	&	5		&	2	&	1	&	5	&  	4	&	164	&	41			&	2	&	2	&	25	\\
	26	&	13		&	1	&	1	&	3	&  	6	&	168	&	21			&	3	&	3	&	729	\\
	28	&	7		&	2	&	1	&	5	&  	4	&	170	&	85			&	1	&	11	&	177147\\	
	30	&	15		&	1	&	3	&	27	&  	8	&	174	&	87			&	1	&	3	&	27	\\
	34	&	17		&	1	&	2	&	9	&  	8	&	178	&	89			&	1	&	4	&	81	\\
	40	&	5		&	3	&	1	&	9	&  	6	&	180	&	45			&	2	&	5	&	3125	\\
	42	&	21		&	1	&	3	&	27	&  	8	&	182	&	91			&	1	&	8	&	6561	\\
	46	&	23		&	1	&	1	&	3	&  	8	&	184	&	23			&	3	&	1	&	9	\\
	50	&	25		&	1	&	2	&	9	&  	4	&	186	&	93			&	1	&	9	&	19683\\	
	52	&	13		&	2	&	1	&	5	&  	6	&	188	&	47			&	2	&	1	&	5	\\
	56	&	7		&	3	&	1	&	9	&  	6	&	190	&	95			&	1	&	3	&	27	\\
	58	&	29		&	1	&	1	&	3	&  	12	&	194	&	97			&	1	&	2	&	9	\\
	60	&	15		&	2	&	3	&	125	&  	8	&	196	&	49			&	2	&	2	&	25	\\
	62	&	31		&	1	&	3	&	27	&  	10	&	200	&	25			&	3	&	2	&	81	\\
	68	&	17		&	2	&	2	&	25	&  	12	&	202	&	101			&	1	&	1	&	3	\\
	70	&	35		&	1	&	4	&	81	&  	14	&	204	&	51			&	2	&	6	&	15625\\	
	74	&	37		&	1	&	1	&	3	&  	12	&	206	&	103			&	1	&	1	&	3	\\
	78	&	39		&	1	&	3	&	27	&  	12	&	208	&	13			&	4	&	1	&	17	\\
	80	&	5		&	4	&	1	&	17	&  	6	&	210	&	105			&	1	&	12	&	531441\\	
	82	&	41		&	1	&	2	&	9	&  	12	&	212	&	53			&	2	&	1	&	5	\\
	84	&	21		&	2	&	3	&	125	&  	10	&	218	&	109			&	1	&	3	&	27	\\
	90	&	45		&	1	&	5	&	243	&  	8	&	220	&	55			&	2	&	3	&	125	\\
	92	&	23		&	2	&	1	&	5	&  	8	&	222	&	111			&	1	&	3	&	27	\\
	94	&	47		&	1	&	1	&	3	&  	12	&	224	&	7			&	5	&	1	&	33	\\
	98	&	49		&	1	&	2	&	9	&  	4	&	226	&	113			&	1	&	4	&	81	\\
	100	&	25		&	2	&	2	&	25	&  	8	&	228	&	57			&	2	&	2	&	25	\\
	102	&	51		&	1	&	6	&	729	&  	&	230	&	115			&	1	&	4	&	81	\\
	104	&	13		&	3	&	1	&	9	&  	&	232	&	29			&	3	&	1	&	9	\\
	106	&	53		&	1	&	1	&	3	&  	&	234	&	117			&	1	&	9	&	19683\\	
	110	&	55		&	1	&	3	&	27	&  	&	238	&	119			&	1	&	7	&	2187	\\
	112	&	7		&	4	&	1	&	17	&  	&	240	&	15			&	4	&	3	&	4913	\\
	114	&	57		&	1	&	2	&	9	&  	&	244	&	61			&	2	&	1	&	5	\\
	116	&	29		&	2	&	1	&	5	&  	&	246	&	123			&	1	&	6	&	729	\\
	120	&	15		&	3	&	3	&	729	&  	&	248	&	31			&	3	&	3	&	729	\\
	122	&	61		&	1	&	1	&	3	&  	&	250	&	125			&	1	&	3	&	27	\\
	124	&	31		&	2	&	3	&	125	&  	&	252	&	63			&	2	&	9	&	1953125\\	
	126	&	63		&	1	&	9	&	19683&  	&	254	&	127			&	1	&	9	&	19683\\	
	130	&	65		&	1	&	6	&	729	&  	&	260	&	65			&	2	&	6	&	15625	\\
	136	&	17		&	3	&	2	&	81	&  	&	266	&	133			&	1	&	7	&	2187	\\
	138	&	69		&	1	&	3	&	27	&  	&	270	&	135			&	1	&	8	&	6561	\\
	140	&	35		&	2	&	4	&	625	&  	&	272	&	17			&	4	&	2	&	289	\\
	142	&	71		&	1	&	1	&	3	&  	&	274	&	137			&	1	&	2	&	9	\\
	146	&	73		&	1	&	4	&	81	&  	&	276	&	69			&	2	&	3	&	125	\\
	148	&	37		&	2	&	1	&	5	&  	&	280	&	35			&	3	&	4	&	6561	\\
	150	&	75		&	1	&	6	&	729	&  	&	282	&	141			&	1	&	3	&	27	\\
	154	&	77		&	1	&	3	&	27	&  	&	284	&	71			&	2	&	1	&	5	\\
	156	&	39		&	2	&	3	&	125	&  	&	286	&	143			&	1	&	3	&	27	\\
     \hline
	\end{tabular}
\end{table}
\end{center}

\noindent
\textbf{\Large Acknowledgment}\\

The authors gratefully acknowledge financial support from the National Research Council of the Philippines.

\end{document}